\documentclass{llncs}
\usepackage{amsmath}
\usepackage{amssymb}
\usepackage{algorithmicx}
\usepackage{algpseudocode}
\usepackage{graphicx}
\usepackage{caption}

\spnewtheorem*{corollary*}{Corollary}{\bf}{\itshape}

\algtext*{EndFunction}
\algtext*{EndProcedure}
\algtext*{EndFor}
\algtext*{EndWhile}
\algtext*{EndIf}
\algloopdefx{NIf}[1]{\textbf{if} #1 \textbf{then}}
\algloopdefx{NElse}{\textbf{else}}
\algloopdefx{NElseIf}{\textbf{else if}}
\algloopdefx{NForAll}[1]{\textbf{for each} #1 \textbf{do}}
\algloopdefx{NWhile}[1]{\textbf{while} #1 \textbf{do}}
\algloopdefx{NFor}[1]{\textbf{for} #1 \textbf{do}}

\title {Online Detection of Repetitions with~Backtracking}

\author
{
	Dmitry Kosolobov
}

\institute{Ural Federal University, Ekaterinburg, Russia\\ \email{dkosolobov@mail.ru}}

\begin{document}

\maketitle

\begin{abstract}
In this paper we present two algorithms for the following problem: given a string and a rational $e > 1$, detect in the online fashion the earliest occurrence of a repetition of exponent $\ge e$ in the string.

1. The first algorithm supports the backtrack operation removing the last letter of the input string. This solution runs in $O(n\log m)$ time and $O(m)$ space, where $m$ is the maximal length of a string generated during the execution of a given sequence of $n$ read and backtrack operations.

2. The second algorithm works in $O(n\log\sigma)$ time and $O(n)$ space, where $n$ is the length of the input string and $\sigma$ is the number of distinct letters. This algorithm is relatively simple and requires much less memory than the previously known solution with the same working time and space.
\end{abstract}

\section{Introduction}

The study of algorithms analyzing different kinds of string periodicities forms an important branch of stringology. Repetitions of a given fixed order often play a central role in such investigations. We say that an integer $p$ is a \emph{period} of $w$ if $w = (uv)^ku$ for some integer $k \ge 1$ and strings $u$ and $v$ such that $|uv| = p$. Given a rational $e > 1$, a string $w$ such that $|w| \ge pe$ for a period $p$ of $w$ is called an \emph{$e$-repetition}. A string is \emph{$e$-repetition-free} if it does not contain an $e$-repetition as a substring. We consider algorithms recognizing $e$-repetition-free strings for any fixed $e > 1$. To be more precise, we say that an algorithm \emph{detects $e$-repetitions} if it decides whether the input string is $e$-repetition-free. Further, we say that this algorithm detects $e$-repetitions \emph{online} if it processes the input string sequentially from left to right and decides whether each prefix is $e$-repetition-free after reading the rightmost letter of that prefix.

In this paper we give two algorithms that detect $e$-repetitions online for a given fixed $e > 1$.
The first one, which uses the ideas of the Apostolico-Breslauer algorithm~\cite{ApostolicoBreslauer}, works on unordered alphabet and supports \emph{backtracking}, the operation removing the last letter of the processed string. This solution requires $O(n\log m)$ time and $O(m)$ space, where $m$ is the maximal length of a string generated during the execution of $n$ given backtrack and read operations. Slightly modifying the proof from~\cite{MainLorentz}, one can show that this time is the best possible in the case of unordered alphabet. The second algorithm works on ordered alphabet and requires $O(n\log\sigma)$ time and linear space, where $\sigma$ is the number of distinct letters in the input string and $n$ is the length of this string. Although this result does not theoretically outperform the previously known solution~\cite{HongChen}, it is significantly less complicated and can be used in practice. Both algorithms report the position of the leftmost \mbox{$e$-repetition}.

Let us point out some previous results on the problem. Recall that a repetition of the form $xx$ is called a \emph{square}. A string is \emph{square-free} if it is $2$-repetition-free. Squares are, perhaps, the most extensively studied repetitions. The classical result of Thue \cite{Thue} states that on a three-letter alphabet there are infinitely many square-free strings. How fast can one decide whether a string is square-free? It turns out that the orderedness of alphabet plays a crucial role here: while any algorithm detecting squares on unordered alphabet requires $\Omega(n\log n)$ time~\cite{MainLorentz}, it is unlikely that \emph{any} superlinear lower bound exists in the case of ordered alphabet, in view of the recent result of the author~\cite{Kosolobov}. So, we always emphasize whether an algorithm under discussion relies on order or not.

The best known offline (not online) results are the algorithm of Main and Lorentz~\cite{MainLorentz} detecting $e$-repetitions in $O(n\log n)$ time and linear space on unordered alphabet, and Crochemore's algorithm~\cite{Crochemore} detecting $e$-repetitions in $O(n\log \sigma)$ time and linear space on ordered alphabets. Our interest in online algorithms detecting repetitions was partially motivated by problems in the artificial intelligence research (see~\cite{LeungPengTing}), where some algorithms use the online square detection. Apostolico and Breslauer \cite{ApostolicoBreslauer} presented a parallel algorithm for this problem on an unordered alphabet. As a by-product, they obtained an online algorithm detecting squares in $O(n\log n)$ time and linear space, the best possible bounds as it was noted above. Later, online algorithms detecting squares in $O(n\log^2 n)$ \cite{LeungPengTing} and $O(n(\log n{+}\sigma))$ \cite{JanssonPeng} time were proposed. Apparently, their authors were unaware of the result of \cite{ApostolicoBreslauer}. For ordered alphabet, Jansson and Peng \cite{JanssonPeng} found an online algorithm detecting squares in $O(n\log n)$ time and Hong and Chen \cite{HongChen} presented an online algorithm detecting $e$-repetitions in $O(n\log\sigma)$ time and linear space.

An online algorithm for square detection with backtracking is in the core of the generator of random square-free strings described in~\cite{Shur}. Using our algorithm with backtracking, one can in a similar way construct a generator of random \mbox{$e$-repetition-free} strings for any fixed $e > 1$. This result might be useful in further studies in combinatorics on words.

The paper is organized as follows. In Section~\ref{SectCatcher} we present some basic definitions and the key data structure, called catcher, which helps to detect repetitions. Section~\ref{SectUnord} contains an algorithm with backtracking. In Section~\ref{SectOrd} we describe a simpler solution without backtracking.

\section{Catcher}\label{SectCatcher}

A \emph{string of length $n$} over the alphabet $\Sigma$ is a map $\{1,2,\ldots,n\} \mapsto \Sigma$, where $n$ is referred to as the length of $w$, denoted by $|w|$. We write $w[i]$ for the $i$th letter of $w$ and $w[i..j]$ for $w[i]w[i{+}1]\ldots w[j]$. Let $w[i..j]$ be the empty string for any~$i > j$. A string $u$ is a \emph{substring} of $w$ if $u = w[i..j]$ for some $i$ and $j$. The pair $(i,j)$ is not necessarily unique; we say that $i$ specifies an \emph{occurrence} of $u$ in $w$. A string can have many occurrences in another string. A substring $w[1..j]$ [resp., $w[i..n]$] is a \emph{prefix} [resp. \emph{suffix}] of $w$. For any $i,j$, the set $\{k\in \mathbb{Z} \colon i \le k \le j\}$ (possibly empty) is denoted by $[i..j]$; $(i..j]$ and $[i..j)$ denote $[i..j] \setminus \{i\}$ and $[i..j] \setminus \{j\}$ respectively.

We fix a rational constant $e > 1$ and use it throughout the paper. The input string is denoted by $text$ and $n = |text|$. Initially, $text$ is the empty string. We refer to the operation appending a letter to the right of $text$ as \emph{read operation} and to the operation that cuts off the last letter of $text$ as \emph{backtrack operation}.

Let us briefly outline the ideas behind our results. Both our algorithms utilize an auxiliary data structure based on a scheme proposed by Apostolico and Breslauer~\cite{ApostolicoBreslauer}. This data structure is called a \emph{catcher}. Once a letter is appended to the end of $text$, the catcher checks whether $text$ has a suffix that is an $e$-repetition of length $k$ such that $k \in [l..r]$ for some segment $[l..r]$ specific for this catcher. The segment $[l..r]$ cannot be arbitrary, so we cannot, for example, create a catcher with $l = 1$ and $r = n$. But, as it is shown in Section~\ref{SectUnord}, we can maintain $O(\log n)$ catchers such that the union of their segments $[l..r]$ covers the whole range from $1$ to $n$ and hence these catchers ``catch'' each $e$-repetition in $text$. This construction leads to an algorithm with backtracking. In Section~\ref{SectOrd} we further reduce the number of catchers to a constant but this solution does not support backtracking.

In what follows we first describe an inefficient version of the read operation for catcher and show how to implement the backtrack operation; then, we improve the read operation and provide time and space bounds for the constructed catcher.

Let $i$ and $j$ be integers such that $1 \le i \le j < n$. Observe that if for some $k \le i$, the string $text[k..n]$ is an $e$-repetition and $e(n - j) \ge n - k + 1$, then the string $text[i..j]$ occurs in $text[i{+}1..n]$ (see Fig.~\ref{fig:catcher}). Given this fact, the read operation works as follows. The catcher searches online occurrences of the string $text[i..j]$ in $text[i{+}1..n]$. If we have $text[i..j] = text[n{-}(j{-}i)..n]$, then the number $p = n - j$ is a period of $text[i..n]$. The catcher ``extends'' the repetition $text[i..n]$ to the left with the same period $p$. Then, the catcher online ``extends'' the repetition to the right with the same period $p$ until an $e$-repetition is found. We say that \emph{the catcher is defined by $i$ and $j$}.
\begin{figure}[htb]
\vskip-4mm
\includegraphics[scale=0.55]{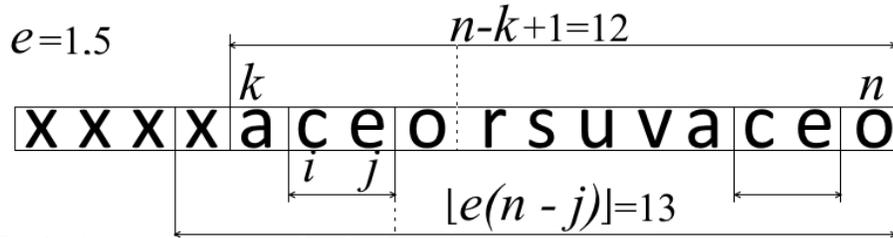}
\vskip-4mm
\caption{\small An $e$-repetition $text[k..n]$, where $k = 5$, $n = 16$. Here $i = 6$, $j = 7$, and $text[i..j] = text[14..15]$.}
\label{fig:catcher}
\vskip-5mm
\end{figure}

\begin{example}
Consider $text = xxxxaceorsuv$. Denote $n = |text|$. Suppose $e = 1.5$. Let a catcher be defined by $i = 6$ and $j = 7$ (see Fig.~\ref{fig:catcher}). We consecutively perform the read operations that append the letters $a, c, e, o$ to the right of $text$. The catcher online searches occurrences of the string $text[i..j] = ce$ (e.g., using the standard Boyer-Moore or Knuth-Morris-Pratt algorithm). Once we have $text = xxxxaceorsuvace$, the catcher has found an occurrence of $text[i..j]$: $text[n{-}1..n] = ce$. Hence, the string $text[i..n] = ceorsuvace$ has a period $p = n - j = 8$. The catcher ``extends'' this repetition to the left and thus obtains the repetition $text[i{-}1..n] = aceorsuvace$ with period $p$. Then the catcher online ``extends'' the found repetition to the right: after the next read operation, the catcher obtains the repetition $text[i{-}1..n] = aceorsuvaceo$ that is an \mbox{$e$-repetition}.
\end{example}

To support the backtrack operation, we store the states of the catcher in an array of states and when the backtracking is performed, we restore the previous state. For the described read operation, this approach has two drawbacks. First, the state does not necessarily require a fixed space, so the array of states may take a large amount of memory. Second, the catcher can spend a lot of time at some text locations (alternating backtracking with reading) and therefore the complexity of the whole algorithm can greatly increase. To solve these problems, our improved read operation performs the ``extensions'' of found repetitions and the searching of $text[i..j]$ simultaneously.

This approach relies on a \emph{real-time constant-space string matching algorithm}, i.e., a constant-space algorithm that processes the input string online, spending constant time per letter; once the searched pattern occurs, the algorithm reports this occurrence. For unordered alphabet, we can use the algorithm of Galil and Seiferas~\cite{GalilSeiferas} though in the case of ordered alphabet, it is more practical to use the algorithm of Breslauer, Grossi, and Mignosi~\cite{BreslauerGrossiMignosi}.

The improved read operation works as follows. Denote $h = (j - i + 1)/2$. The real-time string matching algorithm searches for $text[i..i{+}\lceil h\rceil{-}1]$. It is easy to see that if we have $text[n{-}\lceil h\rceil{+}1..n] = text[i..i{+}\lceil h\rceil{-}1]$, then the number $p = (n - \lceil h\rceil + 1) - i$ is a period of $text[i..n]$. The catcher maintains a linked list $P$ of pairs $(p, l_p)$, where $p$ is found in the described way and $l_p$ is such that $p$ is a period of $text[l_p{+}1..n]$ (initially $l_p = i - 1$). Each read operation tries to extend $text[l_p{+}1..n]$ with the same period $p$ to the right and to the left. If $text[n] \ne text[n{-}p]$, then the catcher removes $(p,l_p)$ from $P$. To extend to the left, we could assign $l_p \gets \min\{l \colon text[l{+}1..n]\text{ has a period }p\}$ but the calculation of this value requires $O(n)$ time while we want to keep within the constant time on each read operation.

In order to achieve this goal, we will extend $r$ symbols to the left after reading a letter. We choose $r = \lceil(e - 1)p / \lfloor h\rfloor\rceil$. Then one of two situations occurs at the moment when $text[i..j] = text[i{+}p..n]$ (i.e.,  an occurrence of $text[i..j]$ is found). Either we have $text[l_p] \ne text[l_p{+}p]$ ($l_p$ cannot be ``extended'' to the left) or $text[l_p{+}1..n]$ is an $e$-repetition. Suppose $text[i..j] = text[i{+}p..n]$ and $text[l_p] = text[l_p{+}p]$. Since at this moment we have performed $\lfloor h\rfloor$ operations decreasing $l_p$ by $r$, we have $l_p = i - 1 - \lfloor h\rfloor r$ and hence $n - l_p \ge p + \lfloor h\rfloor r$. Thus, if we put $r = \lceil(e - 1)p / \lfloor h\rfloor\rceil$, then $n - l_p \ge ep$ and therefore, $text[l_p{+}1..n]$ is an $e$-repetition. The following pseudocode clarifies this description.
\begin{algorithmic}[1]
\State read a letter and append it to $text$ (thereby incrementing $n$)
\State feed the letter to the algorithm searching for $text[i..i{+}\lceil h\rceil{-}1]$
\If{$text[n{-}\lceil h\rceil{+}1..n] = text[i..i{+}\lceil h\rceil{-}1]$} \Comment{found an occurrence}
    \State $p \gets (n - \lceil h\rceil + 1) - i;\;l_p \gets i - 1;$ \Comment{$p$ is a period of $text[l_p{+}1..n]$}
    \State $P \gets P \cup \{(p,l_p)\};$
\EndIf
\ForAll{$(p,l_p) \mathrel{\mathbf{in}} P$}
    \If{$text[n] \ne text[n{-}p]$}
        \State $P \gets P \setminus \{(p,l_p)\};$ \Comment{$text[l_p{+}1..n]$ cannot be ``extended'' to the right}
    \Else
        \State $r \gets \lceil(e - 1)p /\lfloor h\rfloor\rceil;$\label{lst:rvalue} \Comment{maximal number of left ``extensions''}
        \While{$l_p > 0 \mathrel{\mathbf{and}} r > 0 \mathrel{\mathbf{and}} text[l_p] = text[l_p{+}p]$} \label{lst:inloop}
            \State $l_p \gets l_p - 1;\;r \gets r - 1;$ \Comment{``extend'' $text[l_p{+}1..n]$ to the left}\label{lst:inloop2}
        \EndWhile
        \If{$n - l_p \ge ep$} \Comment{if $text[l_p{+}1..n]$ is an $e$-repetition}
            \State detected $e$-repetition $text[l_p{+}1..n]$
        \EndIf
    \EndIf
\EndFor
\end{algorithmic}
A state of the catcher consists of the list $P$ and the state of the string matching algorithm, $O(|P| + 1)$ integers in total. To support the backtracking, we simply store the states of the catcher in an array of states.
\begin{lemma}
Suppose that $i$ and $j$ define a catcher on $text$, $n$ is the current length of $text$, and $c>0$. If the conditions (i) $text[1..n-1]$ is $e$-repetition-free and (ii) $c(j - i + 1)\ge n - i$ hold, then each read or backtrack operation takes $O(c + 1)$ time and the catcher occupies $O((c + 1)(n - i))$ space.\label{CatcherTime}
\end{lemma}
\begin{proof}
Clearly, at any time of the work, the array of states contains $n - i$ states. Each state occupies $O(|P| + 1)$ integers. Hence, to estimate the required space, it suffices to show that $|P| = O(c)$. Denote $v = text[i..i{+}\lceil h\rceil{-}1]$. It follows from the pseudocode that each $(p,l_p) \in P$ corresponds to a unique occurrence of $v$ in $text[i{+}1..n]$. Thus, to prove that $|P| = O(c)$, it suffices to show that the string $v$ has at most $O(c)$ occurrences in $text[i{+}1..n]$ at any time of the work of the catcher. Suppose $v$ occurs at positions $k_1$ and $k_2$ such that $i < k_1 < k_2 < k_1 + |v|$. Hence, the number $k_2 - k_1$ is a period of $v$. Since $text[1..n{-}1]$ is $e$-repetition-free during the work of the catcher, we have $k_2 - k_1 > \frac{1}{e}|v|$. Therefore the string $v$ always has at most $(n - i) / (\frac{1}{e}|v|)$ occurrences in the string $text[i{+}1..n]$. Finally, the inequalities $|v| \ge \frac{1}{2}(j - i + 1)$ and $\frac{n - i}{j - i + 1} \le c$ imply $(n - i) / (\frac{1}{e}|v|) \le 2ec = O(c)$.

Obviously, each backtrack operation takes $O(c)$ time. Any read operation takes at least constant time for each $(p,l_p) \in P$. But for some $(p,l_p) \in P$, the algorithm can perform $O((e - 1)p / h) = O(p/h)$ iterations of the loop in lines~\ref{lst:inloop}--\ref{lst:inloop2} (see the value of $r$ in line~\ref{lst:rvalue}). Since $p \le n - i$ for each $(p,l_p) \in P$, we have $p/h \le 2(n - i) / (j - i + 1) \le 2c$ and therefore, the loop performs at most $O(c)$ iterations. The loop is executed iff $text[l_p] = text[l_p{+}p]$. But since for each $(p,l_p) \in P$, the value of $r$ is chosen in such a way that $text[l_p] = text[l_p{+}p]$ only if $text[i{+}p..n]$ is a proper prefix of $text[i..j]$ (see the discussion above), there are at most $(j - i + 1)/(\frac{1}{e}|v|) \le 2e$ periods $p$ for which the algorithm executes the loop. Finally, we have $O(|P| + 2ec) = O(c)$ time for each read operation.
\qed
\end{proof}

\begin{lemma}
If for some $k$, the string $text[n{-}k{+}1..n]$ is an $e$-repetition and ${n - i < k \le e(n - j)}$, then a catcher defined by $i$ and $j$ detects this repetition.\label{CatcherTrap}
\end{lemma}
\begin{proof}
Let $p$ be the minimal period of $text[n{-}k{+}1..n]$. Since $text[i..j]$ is a substring of $text[n{-}k{+}1..n]$ and $p \le \frac{k}{e} \le n - j$, the string $text[i..j]$ occurs at position $i + p$. Thus, the catcher detects this $e$-repetition when processes this occurrence (see Fig.~\ref{fig:catcher}).
\qed
\end{proof}

We say that a \emph{catcher covers $[l..r]$} if the catcher is defined by integers $i$ and $j$ such that $n - i < n - r + 1 \le n - l + 1 \le e(n - j)$; by Lemma~\ref{CatcherTrap}, this condition implies that if for some $k \in [l..r]$, the suffix $text[k..n]$ is an $e$-repetition, then the catcher detects this repetition. We also say that the catcher \emph{covers a segment of length $r - l + 1$}. Note that if we append a letter to the end of $text$, the catcher still covers $[l..r]$. We say that a set $S$ of catchers covers $[l..r]$ if $\bigcup_{C \in S} [l_C..r_C] \supset [l..r]$, where $[l_C..r_C]$ is a segment covered by catcher $C$.

\section{Unordered Alphabet and Backtracking}\label{SectUnord}

\begin{theorem}
For unordered alphabet, there is an online algorithm with backtracking that detects $e$-repetitions in $O(n\log m)$ time and $O(m)$ space, where $m$ is the length of a longest string generated during the execution of a given sequence of $n$ backtrack and read operations. \label{UnorderedSquares}
\end{theorem}
\begin{proof}
As above, denote $n = |text|$. If $text$ is not $e$-repetition-free, our algorithm skips all read operations until backtrack operations make $text$ $e$-repetition-free. Therefore, in what follows we can assume that $text[1..n{-}1]$ is $e$-repetition-free and thus, all $e$-repetitions of $text$ are suffixes. In our proof we first give an algorithm without backtracking and then improve it to support the backtrack operation.

\textbf{The algorithm without backtracking.} Our algorithm maintains $O(\log n)$ catchers that cover $[1..n{-}O(1)]$ and therefore ``catch'' almost all $e$-repetitions. For each $k \in [0..\log n]$, we have a constant number of catchers covering adjacent segments of length $2^k$. These segments are of the form $(l2^k..(l{+}1)2^k]$ for some integers $l \ge 0$ precisely defined below. Let us fix an integer constant $s$ for which it is possible to create a catcher covering $(n{-}s2^k..n{-}(s{-}1)2^k]$. To show that such $s$ exists, consider a catcher defined by $i = j = n - (s - 1)2^k$. By Lemma~\ref{CatcherTrap}, this catcher covers $(n{-}s2^k..n{-}(s{-}1)2^k]$ iff $e(n - j) = e(s - 1)2^k \ge s2^k$ or, equivalently, $s \ge \lceil\frac{e}{e - 1}\rceil$. As it will be clear below, to make our catchers fast, we must assume that $s > \frac{e}{e-1}$. Note that $s \ge 2$ since $e > 1$, and $s = 2$ implies $e > 2$.

Now we precisely describe the segments covered by our catchers. Denote $t_r = \max\{0, n - ((s - 1)2^r + (n \bmod 2^r))\}$. For any integer $r \ge 0$, $t_r$ is a nonnegative multiple of $2^r$. Let $k \in [0..\log n]$. The algorithm maintains catchers covering the following segments: $(t_{k+1}..t_{k+1} + 2^k], (t_{k+1} + 2^k..t_{k+1} + 2\cdot 2^k], (t_{k+1} + 2\cdot 2^k..t_{k+1} + 3\cdot 2^k], \ldots, (t_k - 2^k..t_k]$ (see Fig.~\ref{fig:system}). Thus, there are at most $\frac{1}{2^k}(t_k - t_{k+1}) \le s$ catchers for each such~$k$. Obviously, the constructed segments cover $[1..n{-}s{+}1]$.
\begin{figure}[htb]
\vskip-4mm
\includegraphics[scale=0.55]{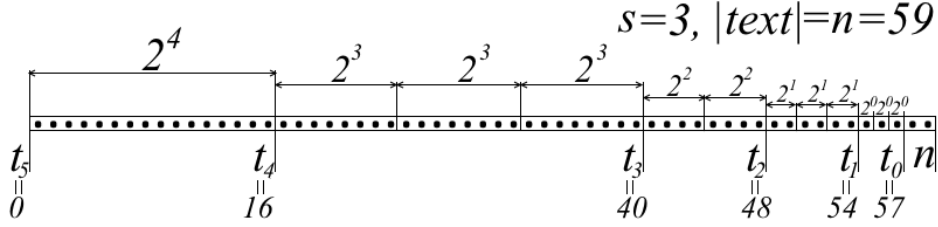}
\vskip-4mm
\caption{\small A system of catchers covering $[1..n{-}s{+}1]$.}
\label{fig:system}
\vskip-5mm
\end{figure}

To maintain this system of catchers, the algorithm loops through all $k \in [0..\log n]$ such that $s2^k \le n$ and, if $n$ is a multiple of $2^k$, creates a new catcher covering $(n - s2^k..n - (s - 1)2^k]$; if, in addition, $n$ is a multiple of $2^{k+1}$, the algorithm removes two catchers covering $(n - s2^{k+1}..n - s2^{k+1} + 2^k]$ and $(n - s2^{k+1} + 2^k..n - (s - 1)2^{k+1}]$. To prove that the derived system covers $[1..n{-}s{+}1]$, it suffices to note that if an iteration of the loop removes two catchers covering $(b_1..b_2]$ and $(b_2..b_3]$, for some $b_1, b_2, b_3$, then the next iteration creates a catcher covering $(b_1..b_3]$. We detect $e$-repetitions of lengths $2, 3, \ldots, s{-}1$ by a simple naive algorithm. In the following pseudocode we use the three-operand $\mathbf{for}$ loop like in the C language.
\begin{algorithmic}[1]
\State read a letter and append it to $text$ (thereby incrementing $n$)
\State check for $e$-repetitions of length $2, 3, \ldots, s{-}1$
\For{$(k \gets 0;\; s2^k \le n \mathrel{\mathbf{and}} n \bmod 2^k = 0;\;k \gets k + 1)$}
    \State create a catcher covering $(n - s2^k..n - (s - 1)2^k]$
    \If{$n \bmod 2^{k+1} = 0 \mathrel{\mathbf{and}} n - s2^{k+1} \ge 0$}
        \State remove the catcher covering $(n - s2^{k+1}..n - s2^{k+1} + 2^k]$
        \State remove the catcher covering $(n - s2^{k+1} + 2^k..n - (s - 1)2^{k+1}]$
    \EndIf
\EndFor
\end{algorithmic}
When the algorithm creates a catcher covering $(n - s2^k .. n - (s - 1)2^k]$, it has some freedom choosing integers $i$ and $j$ that define this catcher. We put $i = n - (s - 1)2^k$ and $j = \max\{i, n - \lceil \frac{s}{e}2^k\rceil\}$. Indeed, in the case $j \ne i$ we have $e(n - j) = e\lceil \frac{s}{e}2^k\rceil \ge s2^k$ and, by Lemma~\ref{CatcherTrap}, the catcher covers $(n - s2^k .. n - (s - 1)2^k]$; the case $j = i$ was considered above when we discussed the value of $s$.

Clearly, the proposed algorithm is correct. Now it remains to estimate the consumed time and space. Consider a catcher defined by integers $i$ and $j$ and covering a segment of length $2^k$. Let us show that $j - i + 1 > \alpha 2^k$ for a constant $\alpha > 0$ depending only on $e$ and $s$. We have $j - i + 1 = (s - 1)2^k - \lceil\frac{s}{e}2^k\rceil + 1 > ((s - 1) - \frac{s}{e})2^k$. The inequality $s > \frac{e}{e - 1}$ implies $(s - 1) - \frac{s}{e} > 0$ (here we use the fact that $s$ is strictly greater than $\frac{e}{e -1}$). Hence, we can put $\alpha = (s - 1) - \frac{s}{e}$.

Denote by $n'$ the value of $n$ at the moment of creation of the catcher. The algorithm removes this catcher when either $n' = n - s2^k$ or $n' = n - (s - 1)2^k$. Thus, since $j - i + 1 > \alpha 2^k$ for some $\alpha > 0$, it follows from Lemma~\ref{CatcherTime} that the catcher requires $O(1)$ time at each read operation and occupies $O(2^k)$ space. Hence, all catchers take $O(s\sum_{k=0}^{\log m} 2^k) = O(m)$ space and the algorithm requires $O(\log m)$ time at each read operation if we don't count the time for creation of catchers. We don't estimate this time in this first version of our algorithm.

\textbf{The algorithm with backtracking.} Now we modify the proposed algorithm to support the backtracking. Denote $n' = n + 1$. The backtrack operation is simply a reversed read operation: we loop through all $k\in [0..\log n']$ such that $s2^k \le n'$ and, if $n'$ is a multiple of $2^k$, remove the catcher covering $(n' - s2^k .. n' - (s - 1)2^k]$; if, in addition, $n'$ is a multiple of $2^{k+1}$, the algorithm creates two catchers covering $(n - s2^{k+1}..n - s2^{k+1} + 2^k]$ and $(n - s2^{k+1} + 2^k..n - (s - 1)2^{k+1}]$. Clearly, this solution is slow: if $n = 2^p$ for some integer $p$, then $n$ consecutive backtrack and read operations require $O(n^2)$ time.

To solve this problem, we make the life of catchers longer. In the modified algorithm, the read and backtrack operations don't remove catchers but mark them as ``removed'' and the marked catchers still work some number of steps. If a backtrack or read operation tries to create a catcher that already exists but is marked as ``removed'', the algorithm just deletes the mark.

How long is the life of marked catcher? Consider a catcher defined by $i = n' - (s - 1)2^k$ and $j = \max\{i, n' - \lceil \frac{s}{e}2^k\rceil\}$, where $n'$ is the value of $n$ at the moment of creation of the catcher in the corresponding read operation. The read operation marks the catcher as ``removed'' when either $n' = n - s2^k$ or $n' = n - (s - 1)2^k$; our modified algorithm removes this marked catcher when $n' = n - (s + 1)2^k$ or $n' = n - s2^k$ respectively, i.e., the catcher ``lives'' additional $2^k$ steps. The backtrack operation marks the catcher as ``removed'' when $n' = n + 1$; we remove this catcher when $n' = n + \min\{2^k, n' - j\}$ (recall that the catcher cannot exist if $n < j$), i.e., the catcher ``lives'' additional $\min\{2^k, \lceil\frac{s}{e}2^k\rceil\} = \Theta(2^k)$ steps.

Let us analyze the time and space consumed by the algorithm. It is easy to see that for any $k\in [0..\log n]$, there are at most $s{+}2$ catchers covering segments of length $2^k$. The worst case is achieved when we have $s$ working catchers and two marked catchers. Now it is obvious that the modified algorithm, as the original one, takes $O(m)$ space and requires $O(\log m)$ time in each read or backtrack operation if we don't count the time for creation of catchers. The key property that helps us to estimate this time is that once a catcher covering a segment of length $2^k$ is created, it cannot be removed during any sequence of $\Theta(2^k)$ backtrack and read operations. To create this catcher, the algorithm requires  $\Theta(2^k)$ time and hence, this time for creation is amortized over the sequence of $\Theta(2^k)$ backtrack and read operations. Thus, the algorithm takes $O(n\log m)$ overall time, where $n$ is the number of read and backtrack operations.
\qed
\end{proof}

\section{Ordered Alphabet} \label{SectOrd}

It turns out that in some natural cases we can narrow the area of $e$-repetition search. More precisely, if $text[1..n{-}1]$ is $e$-repetition-free, then the length of any $e$-repetition of $text$ is close to the length of the shortest suffix $v$ of $text$ such that $v$ does not occur in $text[1..n{-}1]$. In the sequel, $v$ is referred to as the \emph{shortest unioccurrent suffix of $text$}. Denote $t = |v|$. Suppose $u$ is a suffix of $text$ such that $u$ is an $e$-repetition. Let us first consider some specific values of $e$.
\begin{example}
Let $e = 5$. We prove that $t \le |u| < \frac{5}{4}t$. Denote by $p$ a period of $u$ such that $5p \le |u|$. Since the suffix of length $t{-}1$ occurs in $text[1..n{-}1]$ and $text[1..n{-}1]$ is $5$-repetition-free, we have $|u| \ge t$. Suppose, to the contrary, $|u| \ge t + \frac{1}{4}t$. Then $t + p \le t + \frac{1}{5}|u| \le |u|$ and $text[n{-}t{+}1..n] = text[n{-}t{-}p{+}1..n{-}p]$ by periodicity of $u$ (see Fig.~\ref{fig:repet} a), a contradiction to the definition of $t$.
\end{example}
\begin{example}
Let $e = 1.5$. We show that $t \le |u| < \frac{1.5}{0.5}t$. As above, we have $|u| \ge t$. Denote by $p$ a period of $u$ such that $1.5p \le |u|$. Suppose $|u| \ge t + \frac{1}{0.5} t$ (or $t \le \frac{0.5}{1.5}|u|$); then $t + p \le t + \frac{1}{1.5}|u| \le \frac{0.5}{1.5}|u| + \frac{1}{1.5}|u| = |u|$ and $text[n{-}t{+}1..n] = text[n{-}t{-}p{+}1..n{-}p]$ (see Fig.~\ref{fig:repet} b), which contradicts to the definition of $t$.
\end{example}
\begin{figure}[!htb]
\vspace*{-7mm}
\small a\includegraphics[scale=0.50, clip, trim=0 0 0 5]{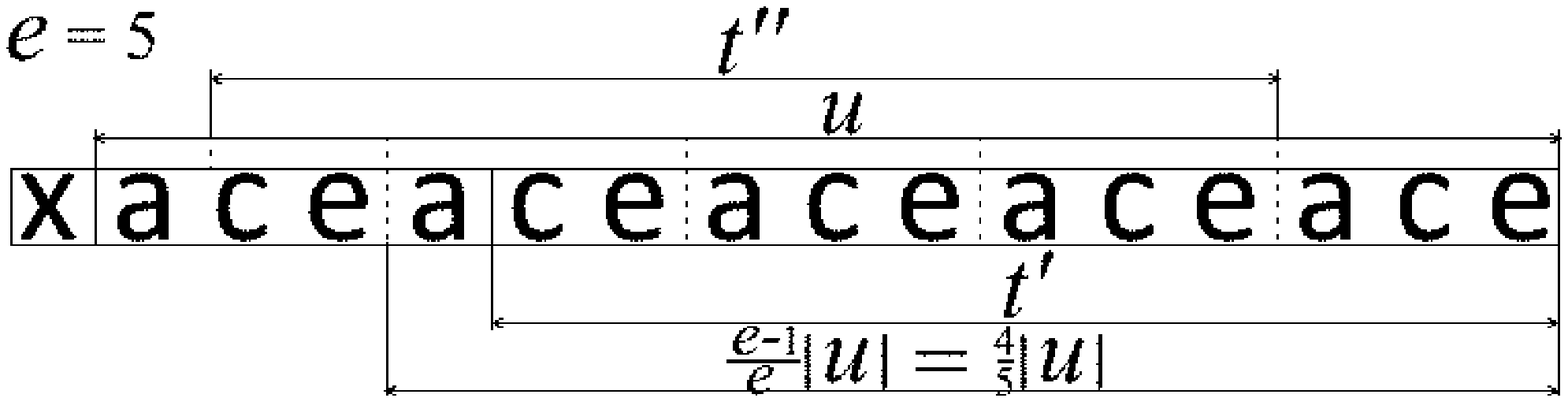}
\small b\includegraphics[scale=0.50, clip, trim=0 0 0 0]{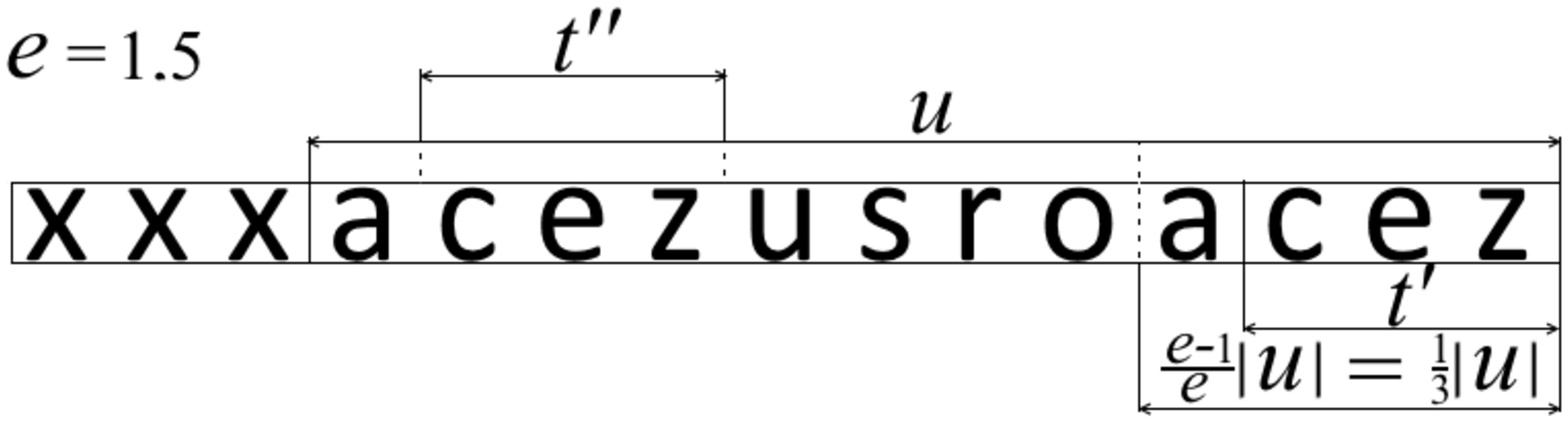}
\vskip-2mm
\caption[]{\small (a) $n = 16$, $u = text[2..n]$, $t = 13$, $t' = 11$, $text[n{-}t'{+}1..n]{=}text[n{-}t'{-}2..n{-}3]$;\\
(b) $n = 15$, $u = text[4..n]$, $t = 5$, $t' = 3$, $text[n{-}t'{+}1..n] = text[n{-}t'{-}7..n{-}8]$.}
\vskip-5mm
\label{fig:repet}
\end{figure}
\begin{lemma}
Let $t$ be the length of the shortest unioccurrent suffix of $text$, and $u$ be an $e$-repetition of $text$. If $text[1..n{-}1]$ is $e$-repetition-free, then $t \le |u| < \frac{e}{e - 1}t$.\label{RepetitionLocation}
\end{lemma}
\begin{proof}
Clearly, $u$ is a suffix. We have $t \le |u|$ since the suffix of length $t{-}1$ occurs in $text[1..n{-}1]$ and $text[1..n{-}1]$ is $e$-repetition-free. Suppose, to the contrary, $|u| \ge \frac{e}{e - 1} t$ (or $t \le \frac{e-1}{e}|u|$). Denote by $p$ the minimal period of $u$. We have $p \le \frac{1}{e}|u|$. Further, we obtain $t + p \le t + \frac{1}{e}|u| \le \frac{e - 1}{e} |u| + \frac{1}{e}|u| = |u|$, i.e., $t + p \le |u|$. Finally, since $p$ is a period of $u$, we have $text[n{-}t{+}1..n] = text[n{-}t{-}p{+}1..n{-}p]$ (see Fig.~\ref{fig:repet} a,b). This contradicts to the definition of $t$.
\qed
\end{proof}

Lemma~\ref{RepetitionLocation} describes the segment in which our algorithm must search $e$-re\-pe\-ti\-tions. To cover this segment by catchers, we use the following technical lemma.
\begin{lemma}
Let $l$ and $r$ be integers such that $0 \le l \le r < n$ and $c(n - r) > n - l$ for a constant $c > 0$. Then there is a set of catchers $\{c_k\}_{k=0}^m$ covering $(l..r]$ such that $m$ is a constant depending on $c$ and each $c_k$ is defined by integers $i_k$ and~$j_k$ such that $j_k - i_k + 1 \ge \frac{e - 1}{2e}(n - r)$.\label{SegmentCover}
\end{lemma}
\begin{proof}
Let us choose a number $\alpha$ such that $0 < \alpha < 1$. Denote $n - r = s$. Consider the following set of catchers $\{c_k\}_{k=0}^m$: $c_k$ is defined by integers $i_k = n - \lceil(e\alpha)^ks\rceil$ and $j_k = n - \lceil\alpha(e\alpha)^ks\rceil$ (see Fig.~\ref{fig:cover}). Denote $i'_k = n - (e\alpha)^ks$ and $j'_k = n - \alpha(e\alpha)^ks$. By Lemma~\ref{CatcherTrap}, $c_k$ covers $(n - e(n - j'_k) .. i'_k] = (n - (e\alpha)^{k+1}s..i'_k]$. Thus, for any $k \in [0..m{-}1]$, the catcher $c_k$ covers $(i'_{k+1}..i'_k]$ and therefore, the set $\{c_k\}_{k=0}^m$ covers the following segment:
$$
(n - (e\alpha)^{m+1}s .. i'_m] \cup (i'_m .. i'_{m-1}] \cup (i'_{m-1} .. i'_{m-2}] \cup \ldots \cup (i'_1..i'_0] = (n - (e\alpha)^{m+1}s..r].
$$
Hence, if $e\alpha > 1$ and $(e\alpha)^{m+1}s \ge cs$, the set $\{c_k\}_{k=0}^m$ covers $(n - cs .. r] \supset (l..r]$. Thus to cover $(l..r]$, we can, for example, put $\alpha = \frac{e + 1}{2e}$ and $m + 1 = \lceil \frac{\log c}{\log(e\alpha)} \rceil = \lceil \frac{\log c}{\log(e + 1) - 1} \rceil$. Finally for $k \in [0..m]$, we have $j_k - i_k + 1 = \lceil(e\alpha)^ks\rceil - \lceil\alpha(e\alpha)^ks\rceil + 1 \ge (e\alpha)^ks - (\alpha(e\alpha)^ks + 1) + 1 = (e\alpha)^k(1 - \alpha)s \ge (1 - \alpha)s = \frac{e - 1}{2e}(n - r)$.
\qed
\end{proof}
\begin{figure}[htb]
\vskip-10mm
\includegraphics[scale=0.55]{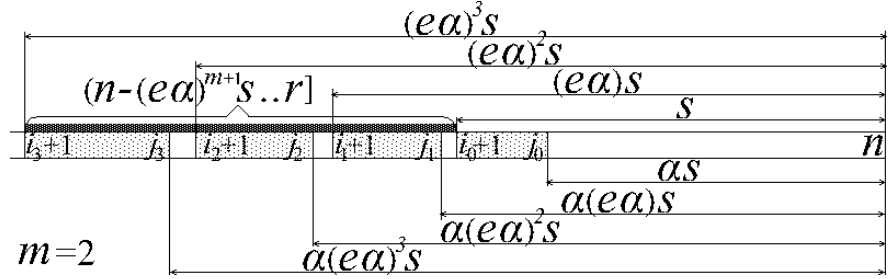}
\vskip-3mm
\caption{\small The system $\{c_k\}_{k=0}^m$ with $m = 2$ ($c_3$ is depicted for clarity), $e \approx 1.5$, $\alpha \approx \frac{5}{6}$.}
\vskip-5mm
\label{fig:cover}
\end{figure}

For each integer $i > 0$, denote by $t_i$ the length of the shortest unioccurrent suffix of $text[1..i]$. We say that there is an online access to the sequence $\{t_i\}$ if any algorithm that reads the string $text$ sequentially from left to right can read $t_i$ immediately after reading $text[i]$. The following lemma describes an online algorithm for $e$-repetition detection based on an online access to $\{t_i\}$. Note that the alphabet is not necessarily ordered.
\begin{lemma}
If there is an online access to the sequence $\{t_i\}$, then there exists an algorithm that online detects $e$-repetitions in linear time and space. \label{OrderedLemma}
\end{lemma}
\begin{proof}
Our algorithm online reads the string $text$ while $text$ is $e$-repetition-free. Let $n = |text|$. Denote $l_n = \max\{0, n - \frac{e}{e - 1}t_n\}$ and $r_n = n - t_n + 1$. By Lemma~\ref{RepetitionLocation}, to detect $e$-repetitions, it suffices to have a set of catchers covering $(l_n..r_n]$. But if the set covers only $(l_n..r_n]$, then we will have to update the catchers in each step $i$ such that $r_{i-1} < r_i$ or $l_{i-1} > l_i$. To reduce the number of updates, we cover $(l_n..r_n]$ with significantly long left and right margins. Thus, some changes of $l_n$ and $r_n$ can be made without rebuilding of catchers.

We maintain two variables $l$ and $r$ such that $l \le l_n \le r_n \le r$. Initially $l = r = 0$. To achieve linear time, we also require $n - r \le 2(r - l)$. The following pseudocode explains how we choose $l$ and $r$:
\begin{algorithmic}[1]
\State read a letter and append it to $text$ (thereby we increment $n$ and read $t_n$)
\State $l_n \gets \max\{0, n - \frac{e}{e - 1}t_n\};\;r_n \gets n - t_n + 1;$
\If{$l_n < l \mathrel{\mathbf{or}} r_n > r \mathrel{\mathbf{or}} n - r > 2(r - l)$} \label{lst:updateCond}
    \State $l \gets \max\{0, n - \frac{2e}{e - 1}t_n\};\;r \gets n - \frac{1}{2}t_n;$
    \State update catchers to cover $(l..r]$\label{lst:update}
\EndIf
\end{algorithmic}
The correctness is clear. Consider the space requirements. Since $n - r = \frac{1}{2}t_n$ and $n - l = \min\{n, \frac{2e}{e - 1}t_n\}$, it follows that $c(n - r) > n - l$ for any $c > 4\frac{e}{e - 1}$. Therefore, by Lemma~\ref{SegmentCover}, the algorithm uses a constant number of catchers and hence requires at most linear space. Denote by $m$ the number of catchers.

Let us estimate the running time. Observe that $r_n$ never decreases. In our analysis, we assume that to increase $r_n$, the algorithm performs $r_n - r_{n-1}$ increments. Obviously, our assumption does not affect the overall running time: to process any string of length $k$, the algorithm executes at most $k$ increments. Also the algorithm performs $k$ increments of $n$. We prove that the time required to maintain catchers is amortized over the sequence of increments of $r_n$ and $n$.

Suppose the algorithm creates a set of catchers $\{c_k\}_{k=1}^m$ at some point. Denote by $n'$ the value of $n$ at this moment. Let us prove that it takes $O(t_{n'})$ time to create this set. For $k \in [1..m]$, let $c_k$ be defined by $i_k$ and $j_k$. By Lemma~\ref{SegmentCover}, for each $k \in [1..m]$, we have $j_k - i_k + 1 \ge \frac{e-1}{2e}(n' - r)$. Since $n' - r \ge \frac{e - 1}{4e}(n' - l) \ge \frac{e - 1}{4e}(n' - i_k)$, we obtain $c(j_k - i_k + 1) \ge n' - i_k$ for any $c \ge 8e^2 / (e - 1)^2$. Hence, by Lemma~\ref{CatcherTime}, it takes $O(n' - i_k)$ time to create the catcher $c_k$. Note that $n' - i_k \le n' - l \le \frac{2e}{e - 1}t_{n'}$ and $\frac{1}{2}t_{n'} \le n' - i_k$, i.e., $n' - i_k = \Theta(t_{n'})$. Therefore, to build the set $\{c_k\}_{k=1}^m$, the algorithm requires $O(\sum_{k=1}^m (n' - i_k)) = O(t_{n'})$ time.

Let us prove that to update the set $\{c_k\}_{k=1}^m$, the algorithm must execute $\Theta(t_{n'})$ increments of $n$ or $r_n$. Consider the conditions of line~\ref{lst:updateCond}:
\begin{enumerate}
\item To satisfy $l_n < l$ (clearly $l > 0$ in this case), since we have $l_{n-1} - l_n \le \frac{e}{e - 1}$ for any $n$, we must perform at least $(l_{n'} - l) / \frac{e}{e-1} = t_{n'}$ increments of $n$.
\item To satisfy $r_n > r$, we must execute $\lceil r - r_{n'}\rceil = \lceil t_n/2\rceil$ increments of $r_n$.
\item To satisfy $n - r > 2(r - l)$, since $n - r = \frac{1}{2}t_{n'} + (n - n')$ and $2(r - l) \ge t_{n'}$, we must increase $n$ by at least $\lceil\frac{1}{2}t_{n'}\rceil$.
\end{enumerate}
The third condition forces us to update catchers after $\lceil\frac{4e}{e - 1}t_{n'}\rceil$ increments of $n$. Indeed, we have $n - r = \lceil\frac{4e}{e - 1}t_{n'}\rceil + n' - r \ge \frac{4e}{e - 1}t_{n'} = 2(n' - l) > 2(r - l)$. Recall that for each $k\in [1..m]$, we have $n' - i_k = \Theta(t_{n'})$ and $j_k - i_k + 1 = \Theta(t_{n'})$. Hence, by Lemma~\ref{CatcherTime}, the catchers $\{c_k\}_{k=1}^m$ take $O(t_{n'})$ overall time. Thus the time required to maintain all catchers is amortized over the sequence of increments of $n$ and $r_n$.
\qed
\end{proof}

\begin{theorem}
For ordered alphabet, there exists an algorithm that online detects $e$-repetitions in $O(n\log\sigma)$ time and linear space, where $\sigma$ is the number of distinct letters in the input string.
\end{theorem}
\begin{proof}
To compute the sequence $\{t_i\}$, we can use, for example, Weiner's online algorithm~\cite{Weiner} (or its slightly optimized version~\cite{BreslauerItaliano}), which works in $O(n\log\sigma)$ time and linear space. Thus, the theorem follows from Lemma~\ref{OrderedLemma}. \qed
\end{proof}

\begin{corollary*}
For constant alphabet, there exists an algorithm that online detects $e$-repetitions in linear time and space.
\end{corollary*}

\noindent\textbf{Acknowledgement.} The author would like to thank Arseny M. Shur for the help in the preparation of this paper and Gregory Kucherov for stimulating discussions.

\bibliographystyle{splncs03}
\bibliography{repetitionFree}

\end{document}